\documentclass[runningheads]{llncs}
\usepackage{amsmath}
\usepackage{stmaryrd}
\usepackage{mathbbol}
\usepackage{yhmath}
\usepackage{url}
\usepackage{times}
\usepackage{latexsym,amsmath,epsfig,amssymb,graphics,tabularx}
\usepackage{color,graphicx}
\usepackage{amsfonts}
\usepackage{graphpap}
\usepackage{listings}
\input{Userdef.sty}

\newcommand{\close}{\emph{cl}}

\title{Extending Hybrid CSP with Probability and Stochasticity}
\author{Yu Peng \and Shuling Wang
\and Naijun Zhan \and Lijun Zhang
}

\institute{State Key Laboratory of Computer Science, Institute of Software, CAS, China}

\begin{document}
\maketitle

\begin{abstract}
  Probabilistic and stochastic behavior are omnipresent in computer
controlled systems, in particular, so-called safety-critical hybrid
systems, because of fundamental properties of nature, uncertain environments,
or simplifications to overcome complexity. Tightly intertwining
discrete, continuous and stochastic dynamics complicates modelling,
analysis and verification of stochastic hybrid systems (SHSs). In the literature,
this issue has been extensively investigated, but unfortunately
it still remains challenging as no promising general solutions are available
yet. In this paper, we give our effort by proposing a general compositional
approach for modelling and verification of SHSs. First, we extend
Hybrid CSP (HCSP), a very expressive and process algebra-like formal
modeling language for hybrid systems, by introducing probability and
stochasticity to model SHSs, which is called stochastic HCSP (SHCSP).
To this end, ordinary differential equations (ODEs) are generalized by
stochastic differential equations (SDEs) and non-deterministic choice is
replaced by probabilistic choice. Then, we extend Hybrid Hoare Logic
(HHL) to specify and reason about SHCSP processes. We demonstrate
our approach by an example from real-world.
\end{abstract}

\section{Introduction}
Probabilistic and stochastic behavior are omnipresent in computer
controlled systems, such as safety-critical hybrid systems, because of
uncertain environments, or simplifications to overcome complexity. For
example, the movement of aircrafts could be influenced by wind; in
networked control systems, message loss and other random effects
(e.g., node placement, node failure, battery drain, measurement
imprecision) may happen.

Stochastic hybrid systems (SHSs) are systems in which discrete,
continuous and stochastic dynamics tightly intertwine. As many of SHSs
are safety-critical, a thorough validation and verification activity
is necessary to enhance the quality of SHSs and, in particular, to
fulfill the quality criteria mandated by the relevant standards. But
modeling, analysis and verification of SHSs is difficult and
challenging.  An obvious research line is to extend hybrid automata
\cite{Henzinger96}, which is the most popular model for traditional
hybrid systems, by adding probability and stochasticity. Then,
verification of SHSs can be done naturally through reachability
analysis, either by probabilistic model-checking
\cite{APLS08,AG97,Bujorianu04,Sproston00,HHWZ10,ZSRHH10,FHHWZ11}, or
by simulation i.e., statistical model-checking \cite{MS06,ZPC13}.
Along this line, several different notions of \emph{stochastic hybrid
  automata} have been proposed
\cite{APLS08,AG97,Bujorianu04,Sproston00,HHWZ10,ZSRHH10,FHHWZ11}, with
the difference on where to introduce randomness. One option is to
replace deterministic jumps by probability distribution over
deterministic jumps. Another option is to generalize differential
equations inside a mode by stochastic differential equations. Stochastic hybrid systems comprising stochastic differential equations
 have been investigated in~\cite{HuLS00,Bujorianu05,APLS08}.  More
general models can be obtained by mixing the above two choices, and by
combining them with memoryless timed probabilistic jumps \cite{BL06},
with a random reset function for each discrete jump \cite{FHHWZ11}. An
overview of this line can be found in \cite{BL06}.

To model complex systems, some compositional modelling formalisms have
been proposed, e.g., HMODEST \cite{HHHK13} and stochastic hybrid
programs \cite{Platzer11}. 
 HCSP due to He, Zhou, et al \cite{He94,ZWR96} is an extension of CSP
\cite{Hoare85} by introducing differential equations to model
continuous evolution and three types of interruptions (i.e.,
communication interruption, timeout and boundary condition) to model
interactions between continuous evolutions and discrete jumps in HSs.
The extension of CSP to probabilistic setting has been investigated by
Morgan et al. \cite{MM96}.  In this paper,
we propose a compositional approach for modelling and verification of
stochastic hybrid systems. First, we extend Hybrid CSP (HCSP), a very
expressive and process algebra-like modeling language for hybrid
systems by introducing probability and stochasticity, called
stochastic HCSP (SHCSP), to model SHSs. In SHCSP, ordinary
differential equations (ODEs) are generalized to stochastic
differential equations (SDEs), and non-deterministic choice is replaced
by probabilistic choice.  Different from Platzer's work~\cite{Platzer11},  SHCSP  provides more expressive constructs for describing hybrid systems, including communication, parallelism, interruption, and so on.

Probabilistic model-checking of SHSs does not scale, in particular,
taking SDEs into account. For example, it is not clear how to
approximate the reachable sets of a simple linear SDEs with more than
two variables.  Therefore, existing verification techniques based on
reachability analysis for SHSs are inadequate, and new approaches are
expected.  As an alternative, in \cite{Platzer11}, Platzer for the first time
investigated how to extend deductive verification to SHSs.   Inspired by Platzer's work, for specifying and reasoning about
SHCSP process, we extend Hybrid Hoare Logic \cite{LLQZ10}, which is an
extension of Hoare logic \cite{Hoare69} to HSs, to SHSs. Comparing
with Platzer's work, more computation features of SHSs, and more expressive constructs such as concurrency,  communication and interruption, can be well
handled in our setting. We demonstrate our approach by modeling and
verification of the example of aircraft planning problem from the
real-world.

%

\section{Background and Notations}

Assume that $\mathcal{F}$  is a $\sigma$-algebra on set $\Omega$ and $P$ is a probability measure on
  $(\Omega, \mathcal{F})$, then $(\Omega, \mathcal{F}, P)$ is called a \emph{probability space}. We here assume that every subset of a null set (i.e., $P(A)=0$) with probability 0 is measurable.
A property which holds with probability 1 is said to hold \emph{almost surely} (\emph{a.s.}).
A \emph{filtration} is a sequence of $\sigma$-algebras $\{\mathcal{F}_t\}_{t\geq 0}$ with $\mathcal{F}_{t_1}\subseteq \mathcal{F}_{t_2}$ for all $t_1<t_2$. We always assume that a filtration $\{\mathcal{F}_t\}_{t\geq 0}$ has been completed to include all null sets and is right-continuous.

Let $\mathcal{B}$ represent the Borel $\sigma$-algebra on  $\mathbb{R}^n$, i.e. the $\sigma$-algebra generated by all open subsets. 
A mapping $X:\Omega \to \mathbb{R}^n$ is called $\mathbb{R}^n$-valued \emph{random variable} if for each $B\in \mathcal{B}$, we have $X^{-1}(B) \in \mathcal{F}$, i.e. $X$ is \emph{$\mathcal{F}$-measurable}. A \emph{stochastic~ process} $X$ is a
function $X:T\times \Omega \to \mathbb{R}^n$  such that for each $t\in T$, $X(t,\cdot):\Omega \to \mathbb{R}^n$ is a random variable, and for each $\omega \in \Omega$, $X(\cdot,\omega):T\to \mathbb{R}^n$ corresponds to a \emph{sample path}. A stochastic process $X$ is \emph{adapted} to a filtration $\{\mathcal{F}_t\}_{t\geq 0}$ if $X_t$ is $\mathcal{F}_t$-measurable. Intuitively, a filtration represents all available historical information of a stochastic process, but nothing related to its future. 
A $c\grave{a}dl\grave{a}g$ function defined on $\mathbb{R}$  is \emph{right continuous} and has \emph{left limit}. A stochastic process $X$ is $c\grave{a}dl\grave{a}g$ iff all of its paths $t\to X_t(\omega)$ (for each $\omega \in \Omega$) are $c\grave{a}dl\grave{a}g$. A $d$-dimensional \emph{Brownian motion} $W$ is a stochastic process with $W_0=0$ that is continuous almost surely everywhere and has independent increments with time, i.e. $W_t-W_s \sim N(0, t-s)~(\mbox{for }0\leq s < t)$, where $N(0, t-s)$ denotes the normal distribution with mean 0 and variance $t-s$. Brownian motion is mathematically extremely complex. Its path  is almost surely continuous everywhere but differentiable nowhere. Intuitively, $W$ can be understood as the limit of a random walk. A \emph{Markov time} with respect to a
 stochastic process $X$ is a random variable $\tau$ such that  for any $t\geq 0$, the event $\{\tau \leq t\}$ is determined by (at most) the information up to time $t$, i.e. $\{\tau \leq t\}\in \mathcal{F}_t$.

We  use \emph{stochastic differential equation} (SDE) to model stochastic continuous evolution, which is of the form $dX_t=b(X_t)dt+\sigma(X_t)dW_t$, where $W_t$ is a Brownian motion.   In which, the drift coefficient $b(X_t)$ determines how the deterministic part of $X_t$ changes with respect to time and the diffusion coefficient $\sigma(X_t)$ determines the stochastic influence to $X_t$ with respect to the Brownian motion $W_t$. Obviously, any solution to an SDE is a stochastic process.

\section{Stochastic HCSP}
\label{sec:Stochastic HybridCSP}

A system in Stochastic HCSP (SHCSP) consists of a finite set of sequential processes in parallel which communicate via channels
synchronously. Each sequential process is represented as a collection of stochastic processes,
each of which arises from the interaction of discrete computation and
stochastic continuous dynamics modeled  by stochastic differential equations.

Let $\textit{Proc}$ represent the set of SHCSP processes, $\Sigma$ the set of channel names. The syntax of SHCSP is given as follows:
 \[
	\begin{array}{lll}
	 P  & ::= & \pskip \mid x :=e  \mid  ch?x \mid ch!e
	           \mid P;Q  \mid B \rightarrow P \mid  P^*\\
	      && \mid  P \sqcup_{p} Q \mid \langle d s = b dt + \sigma dW \& B\rangle\\
	      & & \mid \exempt{\langle d s = b dt + \sigma dW \& B\rangle}{i\in I}{\omega_i \cdot ch_i*}{Q_i} \\[2mm]
	 S & ::= & P \mid S\| S
	 \end{array}
 \]
 Here  \sm{ch,ch_i\in \Sigma}, $ch_i*$ stands for a communication event, e.g. $ch?x$ or $ch!e$,  $x$ is a variable,   \sm{B} and \sm{e} are Boolean and arithmetic expressions, \sm{P, Q, Q_i\in \textit{Proc}} are sequential processes, $p\in [0,1]$ stands for the probability of the choice between $P$ and $Q$, $s$ for a vector of continuous variables, $b$ and $\sigma$ for functions of $s$, $W$ for the Brownian motion process. At the end, $S$ stands for a system, i.e., a SHCSP process.

As defined in the syntax of $P$, the  processes in the first line are original from HCSP, while the last two lines are new for SHCSP. The individual constructs can be understood intuitively as follows:

\begin{itemize}
\item $\pskip$, the assignment $x:=e$, the sequential composition $P; Q$, and the alternative statement $B \rightarrow P$ are defined as usual.
	\item $ch?x$ receives a value along channel $ch$ and assigns it to $x$.
	\item $ch!e$ sends the value of $e$ along channel $ch$.
	A communication takes place when both the sending and the receiving parties are ready, and may cause one side to wait.
		\item The repetition $P^*$ executes $P$ for some finite number of times.
\item $P \sqcup_{p} Q$ denotes probabilistic choice. It behaves as  $P$ with probability $p$ and as $Q$  with probability $1-p$.
	\item $\langle d s = b dt + \sigma dW\& B\rangle$ specifies that the system evolves according to the stochastic process defined by the stochastic differential equation $d s = b dt + \sigma dW$.
	As long as the boolean expression $B$, which defines the
	{\em domain of $s$}, turns false, it terminates. We will later use $d(s)$ to return the dimension of   $s$.
	\item $\exempt{\langle d s = b dt + \sigma dW \& B\rangle}{i\in I}{\omega_i\cdot ch_i*}{Q_i}$ behaves like $ \langle d s = b dt + \sigma dW\& B\rangle$, except that the stochastic evolution is preempted as soon as one of the communications $ch_i*$ takes place, after that the respective $Q_i$ is executed. $I$ is supposed to be finite and for each $i\in I$, $\omega_i \in \mathbb{Q}^+$ represents the \emph{weight} of $ch_i*$. If one or more communications are ready at the same time, say they are $\{ch_j*\}_{j\in J}$ with $J\subseteq I$ and $|J|\geq 1$, then
   $ch_j$ is chosen with the probability $\frac{\omega_j}{\Sigma_{j\in J}\omega_j}$, for each
    $j\in J$. If the stochastic dynamics terminates before a communication among $\{ch_i*\}_I$ occurring, then the process terminates without communicating.
	    \item $S_1\|S_2$ behaves as if $S_1$ and $S_2$ run independently except that all communications along the common channels connecting $S_1$ and $S_2$ are to be synchronized. The processes $S_1$ and $S_2$ in parallel can neither share variables, nor input nor output channels.
\end{itemize}

\subsection{A Running Example}

We use SHCSP to model the aircraft position during the flight, which is inspired from~\cite{Prandini08}. Consider an aircraft that is following a flight path consisting of a sequence of line segments at a fixed altitude. Ideally,  the aircraft should fly at a constant velocity $v$ along the nominal path, but due to the wind or cloud disturbance, the deviation of the aircraft from the path may occur. For safety, the aircraft should follow a correction heading to get back to the nominal path as quickly as possible. On one hand, the correction heading should be orthogonal to the nominal path for the shortest way back, but on the other hand, it should also go ahead to meet the destination. Considering these two objectives, we assume the correction heading  always an acute angle with the nominal path.

Here we model the behavior of the aircraft along one line segment. Without loss of generality, we assume the segment is along $x$-axis, with $(x_s, 0)$ as the starting point and $(x_e, 0)$ as the ending point. When the aircraft deviates from the segment with a vertical distance greater than $\lambda$, we consider it enters a dangerous state. Let $(x_s, y_0)$ be the initial position of the aircraft in this segment, then the future position of the aircraft $(x(t), y(t))$ is governed by the following SDE:
\[
\left(
\begin{array}{c}
  dx(t)\\
  dy(t)
\end{array}\right)
= v
\left(
\begin{array}{c}
  cos(\theta(t))\\
  sin(\theta(t))
\end{array}\right)
dt + dW(t)\]
where $\theta(t)$ is the correction heading and is defined with a constant degree $ \frac{\pi}{4}$ when the aircraft deviates from the nominal path:
\[\theta(t) =
\left\{
\begin{array}{ll}
 -\frac{\pi}{4} \qquad \qquad&\mbox{if } y(t)>0\\
0 \qquad &\mbox{if } y(t)=0\\
 \frac{\pi}{4} \qquad &\mbox{if } y(t)<0
\end{array}\right.\]

Define $B $ be $ x_s\leq x \leq x_e $, the movement of the aircraft described above can be modelled by the following SHCSP process $P_{Air}$:
\[x = x_s; y  = y_0; \langle[dx, dy]^T = v[cos(\theta(t)), sin(\theta(t))]^T dt +
dW(t)\&B\rangle\]

\section{Operational Semantics}
Before giving operational semantics, we introduce some notations first.
\paragraph{\textbf{System Variables}}

	In order to interpret SHCSP processes, we use non-negative reals \sm{\RR^+}  to model time, and introduce a global clock $\now$ as a system variable to record the time in the execution of a process.
	 A \emph{timed communication} is of the form \sm{\pair{ch.c}{b}}, where \sm{ch\in \Chan}, \sm{c \in \RR} and \sm{b \in \RR^+}, representing that a communication along channel \sm{ch} occurs at time \sm{b} with value \sm{c} transmitted. The set \sm{\Chan\times \RR \times \RR^+} of all timed communications is denoted by \sm{\Tevent}. The set of all timed traces is
	  \[\Tevent^*_{\leq} = \{\trace \in \Tevent^* \mid \mbox{ if } \pair{ch_1.c_1}{b_1} \mbox{ precedes }\pair{ch_2.c_2}{b_2} \mbox{ in }\trace, \mbox{ then } b_1\leq b_2\}.\]
         If \sm{C\subseteq \Chan}, \sm{\trace \! \upharpoonright_C} is the projection of \sm{\trace} onto \sm{C} such that only the timed communications along channels of \sm{C} in $\trace$ are preserved.
         Given two timed traces \sm{\trace_1, \trace_2}, and \sm{X\subseteq \Chan},   the \emph{alphabetized parallel} of \sm{\trace_1} and \sm{\trace_2} over \sm{X},
denoted by \sm{\trace_1 \spara \trace_2}, results in the following set  of timed traces
$$\{ \trace \mid
  \trace \! \upharpoonright_{\Chan -(\Chan(\trace_1)\cup \Chan(\trace_2))} =\epsilon, \trace \! \upharpoonright_{\Chan(\trace_1)} = \trace_1, \trace \! \upharpoonright_{\Chan(\trace_2)} = \trace_2 \mbox { and }
    \trace \! \upharpoonright_{X} =  \trace_1 \! \upharpoonright_{X} = \trace_2 \! \upharpoonright_{X}\},$$
     where $\Chan(\trace)$ stands for the set of channels that occur in
 $\trace$.  \oomit{ defined in the following: \\[1mm]
 {\small $\begin{array}{rcl}
    \la \ra \spara \la \ra &\Define&  \{\la \ra \}\\
    \pair{ch.a}{b}\cdot \trace \spara \la \ra &\Define& \left\{ \begin{array}{ll}
        \pair{ch.a}{b}\cdot ( \trace \spara \la \ra )   &\mbox{if } ch \not \in X \\
       \emptyset ~~& \mbox{otherwise}
        \end{array} \right. \\
        \la \ra \spara \trace & \Define & \trace \spara \la \ra \\
    & & \hspace*{-2.8cm} \pair{ch_1.a}{t_1} \cdot \trace_1' \spara \pair{ch_2.b}{t_2} \cdot \trace_2' \\
     &\hspace*{-1cm} \Define & \hspace*{-.6cm}
    \left\{
    \begin{array}{l}
      \pair{ch_1.a}{t_1} \cdot (\trace_1' \spara \trace_2')\  \quad
      ~ \mbox{if \sm{ch_1=ch_2 \in X, a=b, t_1=t_2}}\\
            \pair{ch_1.a}{t_1} \cdot (\trace_1' \spara (\pair{ch_2.b}{t_2} \cdot  \trace_2'))
      \ \cup \pair{ch_2.b}{t_2} \cdot  ((\pair{ch_1.a}{t_1} \cdot \trace_1') \spara \trace_2') \\
      ~~~~ \quad  \mbox{otherwise if \sm{ch_1, ch_2\notin X, t_1= t_2}}\\
      \pair{ch_1.a}{t_1} \cdot ( \trace_1'  \spara(\pair{ch_2.b}{t_2} \cdot  \trace_2'))
      ~~~~ \mbox{otherwise if \sm{ch_1 \notin X, t_1\leq t_2} }\\
      \pair{ch_2.b}{t_2} \cdot ((\pair{ch_1.a}{t_1} \cdot \trace_1') \spara \trace_2')
    ~~~~ \quad   \mbox{otherwise if \sm{ch_2\notin X}, and \sm{t_2\leq t_1} }\\
       \emptyset ~~~ \mbox{otherwise} \\
    \end{array}\right.
       \end{array}$ } \\[1mm]
where we lift the concatenation operator $\cdot$ for traces to a trace set, like in $\pair{ch_1.a}{t_1} \cdot (\trace_1' \spara \trace_2')$, the right of which returns a trace set. }

To model  synchronization of communication events, we need to describe their readiness. Because   a communication itself takes no time when both parties get ready, thus, at a time point, multiple communications may occur. 
In order to record the execution order of communications occurring at the same time point, we prefix each communication readiness a timed trace that happened before
the ready communication event.  Formally, each \emph{communication readiness} has the form of \sm{\gamma.ch?} or \sm{\gamma.ch!}, where $\trace \in \Tevent^*_{\leq}$.
 We denote by $\textit{RDY}$ the set of communication readiness in the sequel.

	Finally, we introduce two system variables, \sm{rdy} and \sm{tr}, to represent the ready set of communication events and the timed  trace accumulated  at the considered time,
respectively. 
In what follows, we use $\textit{Var}(P)$ to represent the set of process variables of $P$, plus the system variables $\{\textit{rdy}, \textit{tr}, \textit{now}\} $ introduced above, which take values respectively from $\RR \cup \textit{RDY} \cup \Tevent^*_{\leq} \cup \RR^+ $, denoted by $\textit{Val}$.

\paragraph{\textbf{States and Functions}}
To interpret a process $P\in \textit{Proc}$,   we define a  state $\textit{ds}$ as a mapping from $\textit{Var}(P)$ to $\textit{Val}$, and denote by $\mathcal{D}$ the set of such states. Because of stochasticity, we introduce a random variable $\rho:\Omega\to \mathcal{D}$ to describe a distribution of all possible states.
 In addition,we introduce  a stochastic process $H: \textit{Intv} \times\Omega\to \mathcal{D}$ to represent the continuous flow of process $P$ over the time interval  $\textit{Intv}$, i.e., state distributions on the interval. In what follows, we will abuse state distribution as state if not  stated otherwise.

 Given two states \sm{\rho_1} and \sm{\rho_2,}  we say \sm{\rho_1} and \sm{\rho_2} are parallelable  iff for each $\omega \in \Omega,$ \sm{\textit{Dom}(\rho_1(\omega))\cap\textit{Dom}(\rho_2(\omega)) =\{\textit{rdy}, tr, \now\}} and \sm{\rho_1(\omega)(\now)=\rho_2(\omega)(\now).} Given two parallelable states \sm{\rho_1} and \sm{\rho_2,}  paralleling them over \sm{X\subseteq \Sigma} results in  a set of new states, denoted by \sm{\rho_1 \uplus \rho_2,}  any of which  \sm{\rho} is given by
      \begin{eqnarray*}
    \rho(\omega)(v)  & \Define & \left\{ \begin{array}{ll}
      \rho_1(\omega)(v) & \mbox{ if } v\in \textit{Dom}(\rho_1(\omega))\setminus \textit{Dom}(\rho_2(\omega)), \\
      \rho_2(\omega)(v) & \mbox{ if } v\in \textit{Dom}(\rho_2(\omega))\setminus \textit{Dom}(\rho_1(\omega)), \\
       \rho_1(\omega)(\now) & \mbox{ if }  v=\now, \\
       \trace, \mbox{ where } \trace \in \rho_1(\omega)(tr) \spara \rho_2(\omega)(tr) & \mbox{ if } v= tr, \\
       \rho_1(\omega)(\textit{rdy}) \cup \rho_2(\omega)(\textit{rdy}) & \mbox{ if } v=\textit{rdy}.
       \end{array}
       \right.
       \end{eqnarray*}
  It makes no sense to  distinguish any two states in
       \sm{\rho_1\uplus \rho_2}, so hereafter we abuse \sm{\rho_1\uplus \rho_2} to represent  any of its elements.\sm{\rho_1\uplus \rho_2} will be used to represent states of parallel processes.

 Given a random variable $\rho$,  the update $\rho[v\to e]$  represents a new random variable such that for any $\omega \in \Omega$ and $x \in \textit{Var}$,
	 $ \rho[v \to e] (\omega)(x)$ is defined as the value of $e$ if $x$ is  $v$, and $\rho(\omega)(x)$ otherwise.
  Given a stochastic process $X  :[0,d)\times\Omega\to R^{d(s)}$, for any $t$ in the domain,  $ \rho[s \to X_t]$ is a new random variable such that  for any $\omega \in \Omega$ and $x \in \textit{Var}$,
  $\rho[s \to X_t] (\omega)(x)$ is defined as  $X(t, w)$ if $x$ is  $s$, and $\rho(\omega)(x)$ otherwise.

	 At last, we define $H_d^\rho$ as the stochastic process over interval $[\rho(now), \rho(now)+d]$ such that for any $t \in [\rho(now), \rho(now)+d]$  and any $\omega$,     $H_d^\rho(t, \omega) = \rho[now \mapsto t](\omega)$, and moreover, $H_d^{\rho, s, X}$ as the stochastic process over interval $[\rho(now), \rho(now)+d] $ such that for any $t \in [\rho(now), \rho(now)+d]$ and any $\omega$,
$H_d^{\rho, s, X}(t, \omega) = \rho[now \mapsto t, \textit{rdy} \mapsto \emptyset, s\mapsto X_t](\omega)$.

\subsection{Operational Semantics}
	Each transition relation has the form of $(P, \rho) \leadm{\alpha} (P', \rho', H)$, where $P$ and $P'$ are processes, $\alpha$ is an event, $\rho, \rho'$ are states, $H$ is a stochastic process. It expresses that starting from initial state $\rho$, $P$ evolves into $P'$ by
performing event $\alpha$, and ends in state $\rho'$ and  the execution history of $\alpha$ is recorded by continuous flow $H$. When the transition is discrete and thus produces a flow on a point interval (i.e. current time $now$), we will write $(P, \rho) \leadm{\alpha} (P', \rho')$ instead of $(P, \rho) \leadm{\alpha} (P', \rho', \{\rho(now) \mapsto \rho'\})$. The label $\alpha$ represents events, which can be an internal event like skip, assignment, or a termination of a continuous \emph{etc}, uniformly denoted by $\tau$, or an external communication event $ch!c$ or $ch?c$, or an internal communication $ch. c$, or a time delay $d$ that is a positive real number. We call the events but the time delay \emph{discrete events}, and will use $\beta$ to range over them. We define the dual of \sm{ch?c} (denoted by $\overline{ch?c}$) as  \sm{ch!c}, and vice versa, and define \sm{\textit{comm}(ch!c, ch?c)} or \sm{\textit{comm}(ch?c, ch!c)} as the communication \sm{ch.c}.
In the operational semantics, besides the timed communications, we will also record the internal events
that have occurred till now in $tr$.

For page limit, we present the semantics for the new constructs of SHCSP in the paper in Table \ref{tab:osemantics}. The semantics for the rest is same to HCSP, which can be found  in Appendix. 
The semantics for probabilistic choice is given by rules (PCho-1) and (PCho-2): it is defined with respect to a random variable $U$ which distributes uniformly in $[0,1]$, such that for any sample $\omega$, if $U(\omega) \leq p $, then $P$ is taken, otherwise, $Q$ is taken. In either case, it is assumed that an internal action happened.
A stochastic dynamics can continuously evolve for \sm{d} time units if \sm{B} always holds during this period, see (Cont-1). In (Cont-1), the variable $X$ solves the stochastic process and the ready set keeps unchanged, reflected by the flow $H_d^{\rho, s, X}$.  The stochastic dynamics terminates at a point whenever \sm{B} turns out false at a neighborhood of the point  (Cont-2). Communication interrupt evolves for \sm{d} time units if none of the communications \sm{ch_i*} is ready (IntP-1), or is interrupted to execute \sm{ch_{i_j}*} whenever \sm{ch_{i_j}*} occurs first (IntP-2), or terminates immediately in case the continuous terminates before  any communication happening (IntP-3).

The following theorem indicates that the semantics of SHCSP  is well defined.
\begin{theorem}
%
 For each transition $(P, \rho) \leadm{\alpha} (P', \rho', H)$, $H$ is an almost surely $c\grave{a}dl\grave{a}g$ process and adapted to the completed filtration $(\mathcal{F}_t)_{t\geq 0}$ (generated by $\rho$, the Brownian motion $(B_s)_{s\leq t}$, the weights $\{\omega_i\}_{i\in I}$ and uniform $U$ process)  and the evolving time from $P$ to $P'$, denoted by  $\Delta(P,P')$, is a $Markov~ time$.
   \label{welldefined}
\end{theorem}
\begin{proof}
   The proof of this theorem can be found  in Appendix.  
\end{proof}

\begin{table}[!ht]
\begin{align}
&{ \fracN{U \mbox{ is a random variable distributed uniformly in }[0,1], \, U(\omega)\le p}
	{(P\sqcup_{p} Q, \rho) \leadm{\tau} \left(P, \begin{array}{r}\rho[tr \mapsto tr \cdot \langle \tau, now\rangle]
	 \end{array}\right)} }& \tag{PCho-1}
	 \\[1mm]
&	{ \fracN{U \mbox{ is a random variable distributed uniformly in }[0,1], \, U(\omega)  > p}
	{(P\sqcup_{p} Q, \rho) \leadm{\tau} \left(Q, \begin{array}{r}\rho[tr \mapsto tr \cdot \langle \tau, now\rangle]
	 \end{array}\right)} }& \tag{PCho-2}
	 \\[1mm]
&	 {\fracN{\begin{array}{l}
	  X :[0,d)\times\Omega\to \mathbb{R}^{d(s)}\mbox{ is the solution of } \\
	   d s = b dt + \sigma dW  \wedge \forall t\in[0,d), \forall \omega. \rho[\now\mapsto \now+t, s \mapsto X_t](\omega)(B) = \ptrue
	   \end{array}}
	  {(\la \textit{ds}=bdt+\sigma dW \& B \ra,
	    \rho) \leadm{d}
	   \left( \begin{array}{l} \la \textit{ds}=bdt+\sigma dW \& B \ra, \\
	      \rho[\now \mapsto \now+d,
	    s \mapsto X_d],
	      H_d^{\rho, s, X}
	     \end{array} \right)} }& \tag{Cont-1}
	 \\[1mm]
&	 {\fracN{ \begin{array}{l}
	 \exists \omega. (\rho(\omega)(B) =   \pfalse) \mbox{ or }
	  ( X :[0,d)\times\Omega\to \mathbb{R}^{d(s)}\mbox{ is the solution of }  \textit{ds}=bdt+\sigma dW ,\\
	   \exists \varepsilon>0
	   \forall t \in (0,\varepsilon)\exists \omega. \rho[\now\mapsto \now+t,s\mapsto X_t](\omega)(B) = \pfalse)
	  \end{array}}
	  {\begin{array}{l} (\la \textit{ds}=bdt+\sigma dW \& B \ra, \rho) \leadm{\tau}
	   (\pstop, \rho[ tr \mapsto tr \cdot \langle \tau, now\rangle )
	     \end{array}} }& \tag{Cont-2}
	  \\[1mm]
&	  \fracN{
	  \begin{array}{l}
	     (ch_i*;Q_i, \rho) \leadm{d} (ch_i*; Q_i, \rho_i', H_i), ~~  \forall i \in I  \\
	  (\la \textit{ds}=bdt+\sigma dW \& B \ra,\rho) \leadm{d} (\la \textit{ds}=bdt+\sigma dW \& B \ra, \rho', H)
	  \end{array}}
	  {     \begin{array}{l}  (\exempt{\la \textit{ds}=bdt+\sigma dW \& B \ra}{i\in I}{\omega_i \cdot ch_i*}{Q_i},
	   \rho)
	   \leadm{d}  \\
	   ~~~~~~~~~~~~~~~~~~~~~~~~~~~~  \left( \begin{array}{l}\exempt{\la \textit{ds}=bdt+\sigma dW \& B \ra}{i\in I}{\omega_i \cdot ch_i*}{Q_i}, \\
	    \rho'[\textit{rdy} \mapsto \cup_{i\in I}\rho_i'(\textit{rdy})], H[\textit{rdy} \mapsto \cup_{i\in I}\rho_i'(\textit{rdy})]
	   \end{array}\right)
	   \end{array}}& \tag{IntP-1}
	   \\[1mm]
&\fracN{
	  \begin{array}{l}
	   \{\overline{ch_{i_k}*}\}_{1\leq k \leq n} \mbox{ get ready simultaneously while others not}\\
	   U \mbox{ is a random variable distributed uniformly in [0,1]},  \mbox{ and for } 1 \leq j \leq n\\
	   \frac{\sum_{k=1}^{j-1}\omega_{i_k}}{\sum_{k=1}^{n}\omega_{i_k}} \leq U(\omega)< \frac{\sum_{k=1}^{j}\omega_{i_k}}{\sum_{k=1}^{n}\omega_{i_k}} \mbox{ and }
	     ( ch_{i_j}*; Q_{i_j},\rho) \leadm{ch_{i_j}*} (Q_{i_j}, \rho')
	  \end{array}}
	  { (\exempt{\la \textit{ds}=bdt+\sigma dW \& B \ra}{i\in I}{\omega_i \cdot ch_i*}{Q_i},
	   \rho) \leadm{ch_{i_j}*}
	   ( {Q_{i_j}},
	    \rho')}& \tag{IntP-2}
	  \\[1mm]
	&   \fracN{
	  \begin{array}{l}
	  (\la \textit{ds}=bdt+\sigma dW \& B \ra, \rho) \leadm{\tau} (\pstop, \rho')
	  \end{array}}
	  {  (\exempt{\la \textit{ds}=bdt+\sigma dW \& B \ra}{i\in I}{\omega_i \cdot ch_i*}{Q_i},
	   \rho) \leadm{\tau}
	   ( \pstop, \rho')}& \tag{IntP-3}\\ \nonumber
\end{align}
\caption{The semantics of new constructs of SHCSP}
\label{tab:osemantics}
\end{table}

\section{Assertions and Specifications}
\label{sec:Specifications}
	In this section, we define a specification logic for reasoning about SHCSP programs. We will first present the assertions including syntax and semantics, and then the specifications based on Hoare triples. The proof system will be given in next section.

\subsection{Assertion Language}
\label{subsec:assertions}
The  assertion language  is essentially defined by a first-order logic with emphasis on the notion of explicit time and the addition of several specific predicates on occurrence of communication traces and events. Before giving the syntax of assertions, we introduce three kinds of expressions first.
\[
\begin{array}{lll}
h & ::= & \varepsilon  \mid \pair{ch.E}{T} \mid h \cdot h \mid h^*\\
 E & ::= & c  \mid x  \mid f^k(E_1, ..., E_k)\\
 T & ::= & o  \mid now \mid u^l(T_1, ..., T_l)
\end{array}\]
$h$ defines trace expressions, among which $\pair{ch.E}{T}$ represents that there is a value $E$ transmitted along channel $ch$ at time $T$. $E$ defines value expressions, including a value constant $c$, a variable $x$, or arithmetic value expressions. $T$ defines time expressions, including a time constant $o$, system variable $now$, or arithmetic time expressions.

The  categories of the assertion language include terms, denoted by $\termt, \termt_1 $ \emph{etc.}, state formulas, denoted by $S, S_1$ \emph{etc.},  formulas, denoted by $\formt, \formt_1$ \emph{etc.}, and probability formulas, denoted by $\mathcal{P}$  \emph{etc.}, which are given by the following BNFs:

\[\begin{array}{lll}
\termt & ::= & E  \mid  T \mid h \mid tr\\
S & ::= &  \bot \mid R^n(\termt_1, ..., \termt_n) \mid   h.ch? \mid h.ch! \mid \neg S \mid S_1 \vee S_2 \\
\formt &::=& \bot  \mid \rdyat{S}{T} \mid \neg \formt \mid \formt_1 \vee \formt_2 \mid \forall v. \formt \mid \forall t. \formt\\
\mathcal{P} &::=& P(\formt)\bowtie p \mid \neg \mathcal{P} \mid \mathcal{P} \vee \mathcal{P}
\end{array}\]
The terms $\termt$ include value, time and trace expressions, plus trace variable $tr$. The state expressions $S$  include false (denoted by $\bot$), truth-valued relation $R^n$ on terms, readiness, and logical combinations of state formulas. In particular, the readiness  $h.ch?$ or $h.ch!$ represents that the communication event $ch?$ or $ch!$ is enabled, and prior to it, the sequence of communications recorded in $h$ has occurred. The formulas $\varphi$ include false, a  primitive $\rdyat{S}{T}$ representing that $S$ holds at time $T$; and logical combinations of formulas ($v,t$ represent logical variables for values and time resp.). For time primitive, we have an axiom that $(\rdyat{S_1}{T}\wedge \rdyat{S_2}{T}) \equi \rdyat{(S_1\wedge S_2)}{T}$. We omit all the other axiom and inference rules for the formulas, that are same to first-order logic. The probability formula $\mathcal{P}$ has the form $P(\formt)\bowtie p$, where $\bowtie \in \{<, \leq, >, \geq \}$, $p\in \mathbb{Q}\cap [0,1]$, or the logical composition of probability formulas free of quantifiers. In particular,  $P(\formt)\bowtie p$ means that  $\formt$ is true with probability $\bowtie p$. For the special case $P(\formt)=1$, we write $\formt$ for short.

In the sequel, we use the standard logical abbreviations, as well as
\[
\begin{array}{c}
  \rdyduring{\formt}{[T_1,T_2]}  \!\Define  \!\forall t. (T_1 \!\leq t\!\leq T_2) \!\Rightarrow \!\rdyat{\formt}{t} \\
  \rdyin{\formt}{[T_1, T_2]}  \!\Define  \exists t. (T_1 \!\leq t\!\leq T_2) \!\wedge  \!\rdyat{\formt}{t}
\end{array}\]

\paragraph{\textbf{Interpretation}}
In the following, we will  use a random variable $Z:\Omega \to (\textit{Var} \to \textit{Val})$ to describe the current state and a stochastic process $\timefun: [0, +\infty) \times \Omega \to (\textit{Var} \to \textit{Val})$ to represent the whole evolution. 
The semantics of a term $\theta$ is a function $\newsem{\theta}: (\Omega \to (\textit{Var}  \to \textit{Val})) \to (\Omega \to \textit{Val})$ that maps any random variable $Z$ to a random variable $\newsem{\theta}^Z$, defined as follows:
\[\begin{array}{l}
\newsemz{c}=c \\
\newsemz{x}=Y \mbox{ where } Y(\omega)=Z(\omega)(x) \mbox{ for } \omega \in \Omega \\
\newsemz{f^k(E_1, ..., E_k)}=f^k(\newsemz{E_1}, ..., \newsemz{E_k}) \\
\newsemz{o}=o \\
\newsemz{now}=Y \mbox{ where } Y(\omega)=Z(\omega)(now) \mbox{ for } \omega \in \Omega \\
\newsemz{u^l(T_1, ..., T_l)}=u^l(\newsemz{T_1}, ..., \newsemz{T_l}) \\
\newsemz{\varepsilon}=\varepsilon \\
\newsemz{\pair{ch.E}{T}}=\pair{ch.\newsemz{E}}{\newsemz{T}} \\
\newsemz{h_1 \cdot h_2}=\newsemz{h_1} \cdot \newsemz{h_2} \\
\newsemz{h^*}= (\newsemz{h})^*
\end{array}\]
The semantics of state formula $S$ is a function $\newsem{S}: (\Omega \to (\textit{Var}  \to \textit{Val})) \to (\Omega \to \{0, 1\})$ that maps any random variable $Z$ describing the current state to a boolean random variable $\newsemz{S}$, defined as follows:
\[\begin{array}{l}
\newsemz{\bot}=0 \\
\newsemz{R^n(\termt_1, \dots, \termt_n)}= R^n(\newsemz{\termt_1}, \dots, \newsemz{\termt_n})\\
\mbox{where } R^n(\newsemz{\termt_1}, \dots, \newsemz{\termt_n})(\omega) =
  R^n(\newsemz{\termt_1}(\omega), \dots, \newsemz{\termt_n}(\omega))\\
\newsemz{h.ch?}=\mathcal{I}_{\{\omega \in \Omega |\newsemz{h}(\omega).ch? \in Z(\omega)(rdy) \}} \\
\newsemz{h.ch!}=\mathcal{I}_{\{\omega \in \Omega |\newsemz{h}(\omega).ch! \in Z(\omega)(rdy) \}} \\
\newsemz{\neg S}= 1-\newsemz{S} \\
\newsemz{S_1 \vee S_2} = \newsemz{S_1} + \newsemz{S_2} - \newsemz{S_1} * \newsemz{S_2}
\end{array}\]
where given a set $S$,  the characteristic function $\mathcal{I}_S$ is defined such that $\mathcal{I}_S(w) =1$ if $w \in S$ and $\mathcal{I}_S(w) = 0$ otherwise. The semantics of formula $\formt$ is interpreted over a stochastic process and  an initial random variable. More precisely, it's a function $\newsem{\formt}: ([0, +\infty)\times \Omega \to (\textit{Var}  \to \textit{Val})) \to (\Omega \to (\textit{Var}  \to \textit{Val})) \to  (\Omega \to \{0, 1\}) $ that maps a stochastic process $\timefun$ with initial state $Z$ to a boolean random variable $\newsemhz{\formt}$. The definition is given below:
\[\begin{array}{l}
\newsemhz{\bot}=0 \\
\newsemhz{\rdyat{S}{T}}= \Lbrack S \Rbrack^{\timefun(\newsemz{T})}\\
\newsemhz{\neg \formt} =1- \newsemhz{\formt} \\
\newsemhz{\formt_1 \vee \formt_2}=\newsemhz{\formt_1}+\newsemhz{\formt_2}- \newsemhz{\formt_1}*\newsemhz{\formt_2}\\
\newsemhz{\forall v.\formt }=\mbox{inf}\{\newsemhz{\formt[b/v]} :b \in \mathbb{R} \}  \\
\newsemhz{\forall t.\formt }=\mbox{inf}\{\newsemhz{\formt[b/t]} : b\in \mathbb{R^+} \}
\end{array}\]
The semantics of probability formula $\mathcal{P}$ is defined by function $\newsem{\mathcal{P}}: ([0, +\infty)\times \Omega \to (\textit{Var}  \to \textit{Val})) \to (\Omega \to (\textit{Var}  \to \textit{Val})) \to   \{0, 1\} $ that maps a stochastic process $\timefun$ with initial state $Z$ to a boolean variable $\newsemhz{\mathcal{P}}$. Formally,
\[\newsemhz{P(\formt)\bowtie p} = (P(\newsemhz{\formt}= 1)= P(\{\omega \in \Omega :\newsemhz{\formt}(\omega)= 1\}) \bowtie p)\]
The semantics for $\neg$ and $\vee$ can be defined as usual.

We have proved that the terms and formulas of the assertion language are measurable, stated by the following theorem:
\begin{theorem}[Measurability]
  For any random variable $Z$ and any stochastic process $\timefun$, the semantics of $\newsemz{\theta}$, $  \newsemz{S}$ and $\newsemhz{\formt}$ are random variables (i.e. measurable).
  \label{Them:Mea}
\end{theorem}
\begin{proof}
   The proof of this theorem is given in Appendix.
\end{proof}

\subsection{Specifications}
Based on the assertion language, the specification for a SHCSP process $P$ is defined as a Hoare triple of the form $\triple{A; E}{P}{R; C}$, where $A, E, R, C$ are probability formulas. $A$ and $R$ are \textit{precondition} and \textit{postcondition}, which specify the initial state and the terminating state of $P$  respectively. For both of them, the formulas $\formt$ occurring in them have the special form $\rdyat{S}{now}$, and we will write  $S$ for short.
$E$ is called an \emph{assumption} of $P$, which expresses the timed occurrence of the dual of communication events provided by the environment.
 $C$ is called a \emph{commitment} of $P$, which expresses the timed occurrence of communication events, and the real-time properties of  $P$.

\begin{definition}[Validity]
We say a Hoare triple \sm{\triple{A; E}{P}{R; C}} is \emph{valid}, denoted by \sm{\models \triple{A; E}{P}{R; C},}
iff for any process \sm{Q,}  any initial states \sm{\funI_1} and \sm{\funI_2,} if  \sm{P} terminates, i.e.\sm{(P \| Q, \funI_1 \uplus \funI_2) \leadm{\alpha^*} (\pstop\|Q', \funI_1' \uplus \funI_2',\timefun)}
then \sm{\Lbrack A \Rbrack^{\funI_1}} and \sm{\Lbrack E \Rbrack^{\timefun,\funI_2}} imply
\sm{\Lbrack R \Rbrack^{\funI_1'}} and
\sm{\Lbrack C \Rbrack^{\timefun,\funI_1'},} where  $\timefun$ is the stochastic process of the evolution.
\end{definition}

%

\section{Proof System}
\label{sec:Proofsystem}
We present a proof system for reasoning about all valid Hoare triples for SHCSP processes. First we axiomatize SHCSP language by defining the axioms and inference rules for all the primitive  and compound constructs, and then the general rules and axioms that are applicable to all processes.



\prag{Skip} The rule for skip is very simple. Indicated by $\top$, the skip process requires nothing from the environment for it to execute, and guarantees nothing during its execution.
\[\triple{A; \top}{\pskip}{A; \top}\]

\prag{Assignment}
The assignment $x:=e$ changes nothing but assigns $x$ to $e$ in the final state, taking no time to complete. \[\triple{A[e/x]; \top}{x:=e}{A; \top}\]

\prag{Input}
For input $ch?x$, we use logical variables $o$ to denote the starting time, $h$  the initial trace, and $v$ the initial value of $x$ respectively,  in the precondition. The assumption indicates that the compatible output event is not ready during $[o, o_1)$, and at time $o_1$, it becomes ready. As a consequence of the assumption, during the whole interval $[o, o_1]$, the input event keeps waiting and ready, as indicated by the commitment. At time $o_1$, the communication occurs and terminates immediately. As indicated by the postcondition, $x$ is assigned by some value $v'$ received,  the trace is augmented by the new pair $\pair{ch.v'}{o_1}$, and $now$ is increased to $o_1$.
Assume $A$ does not contain $tr$ and $o_1$ is finite (and this assumption will be adopted for the rest of the paper). Let $h'$ be $h[v/x,o/now] \cdot \pair{ch.v'}{o_1}$, the rule is presented as follows:
\[\ltriple{A \wedge now = o \wedge tr=h \wedge x=v ; \rdyduring{\neg h.ch!}{[o,o_1)} \wedge \rdyat{h.ch!}{o_1}}{ch?x}
{\{A[o/now] \wedge now = o_1 \wedge \exists v'. (x=v' \wedge tr=h'); \rdyduring{h.ch?}{[o, o_1]}\}}\]
A communication event is equivalent to a sequential composition of a wait statement and an assignment, both of
which are deterministic. Thus, as shown above, the formulas related to traces and readiness hold with probability 1.

If such finite $o_1$ does not exist, i.e., the compatible output event will never become available.
As a consequence, the input event will keep waiting forever, as shown by the following rule:
\[\ltriple{A \wedge now = o \wedge tr=h; \rdyduring{\neg h.ch!}{[o, \infty)}}{ch?x}
{\{A[o/now] \wedge now = \infty; \rdyduring{h.ch?}{[o, \infty)}\}}\]

\prag{Output}
Similarly, for output $ch!e$, we have one rule for the case when the compatible input event becomes ready in finite time. Thus the communication occurs successfully.
 \[\ltriple{A \wedge now = o \wedge tr=h; \rdyduring{\neg h.ch?}{[o,o_1)} \wedge \rdyat{ h.ch?}{o_1}}{ch!e}
{\{A[o/now] \wedge now = o_1 \wedge tr = h[o/now] \cdot \pair{ch.e}{o_1},  \rdyduring{ h.ch!}{[o, o_1]}\}}\]
We also have another rule for the case when the compatible input event will never get ready.
\[\ltriple{A \wedge now = o \wedge tr=h; \rdyduring{ (\neg h.ch?)}{[o,\infty)} }{\, ch!e\, }
{\{A[o/now] \wedge now = \infty; \rdyduring{h.ch!}{[o,\infty)}\} } \]

\prag{Stochastic Differential Equation}
Let $f$ be a function, and $\lambda > 0, p \geq 0$ are real values. We have the following rule for $\langle d s = b dt + \sigma dW \& B\rangle$.
\[\fracN{\begin{array}{c}
  f(s) \in C^2(\mathbb{R}^n, \mathbb{R}) \mbox{ has compact support on } B,
   \lambda,p > 0\mbox{ and }  \\ A \rightarrow B\to (f\leq \lambda p)  \quad B \to (f\geq 0) \wedge( Lf \leq 0)
\end{array}}
{\begin{array}{c}
 \{A \wedge s=s_0 \wedge now=o; \top\}\langle d s = b dt + \sigma dW \& B\rangle
      \{P(f(s)\geq \lambda) \leq p \wedge A[s_0/s, o/now] \\
      \wedge now=o+d \wedge \close(B); B \wedge P(\rdyduring{f(s)\geq \lambda  }{[o, o+d]})\leq p\}
\end{array}
}\]
where $o, s_0$ are logical variables denoting the starting time and the initial value of $s$ resp., $d$ is the execution time of the SDE, and $\close(B)$ returns the closure of $B$, e.g. $\close(x<2) = x\leq 2$;  and the Lie derivative $Lf(s)$ is defined as $\mathop{\sum}\limits_{i}b_i(s)\frac{\partial f}{\partial s_i}(s)+\frac{1}{2}\mathop{\sum}\limits_{i,j}(\sigma (s)\sigma(s)^T)_{i,j}\frac{\partial^2 f}{\partial s_i\partial s_j}(s)$. The rule states that, if the initial state of the SDE satisfies $f \leq \lambda p$, and in the domain $B$, $f$ is always non-negative and $Lf$ is non-positive, then during the whole evolution of the SDE, the probability of $f(s) \geq \lambda$ is less than or equal to $p$; on the other hand, during the evolution, the domain $B$ holds almost surely, while at the end, the closure of $B$ holds almost surely.

\prag{Sequential Composition}
For $P; Q$, we use $o$ to denote the starting time, and $o_1$ the termination time of $P$, if $P$ terminates, which is also the starting time of $Q$. The first rule is for the case when $P$ terminates.
\[\fracN{\begin{array}{c}
  \triple{A \wedge now=o; E}{P}{R_1 \wedge now = o_1; C_1} ~~
 \triple{R_1 \wedge now=o_1; C_1}{Q}{R; C}
\end{array}}
{\triple{A; E}{P; Q}{R; C}}\]
On the other hand, if $P$ does not terminate, the effect of executing $P; Q$ is same to that of executing $P$ itself.
\[\fracN{\begin{array}{c}
  \triple{A \wedge now=o; E}{P}{R \wedge now = \infty; C}
\end{array}}
{\triple{A \wedge now=o; E}{P; Q}{R \wedge now = \infty; C}}\]

\prag{Conditional}
There are two rules depending on whether $B$ holds or not initially.
\[\begin{array}{lcr}
 \fracN{A \Rightarrow B \quad \triple{A; E}{P}{R; C}}{\triple{A; E}{B \rightarrow P}{R; C}}
  & \mbox{ and } &
\fracN{A \Rightarrow \neg B }{\triple{A; \top}{B \rightarrow P}{A; \top}}
\end{array}
\]

\prag{Probabilistic Choice}
The rule for $P  \sqcup_p Q$ is defined as follows:
\[\fracN{
\begin{array}{c}
\triple{A \wedge now=o; E}{P}{P(S) \bowtie_1 p_1; P(\formt) \bowtie_2 p_2}  \\
\triple{A \wedge now=o; E}{Q}{P(S) \bowtie_1 q_1; P(\formt) \bowtie_2 q_2}
\end{array}}
{\triple{A \wedge now=o; E}{~ P  \sqcup_p Q~ }{P(S) \bowtie_1 pp_1+(1-p)q_1; P(\formt) \bowtie_2 pp_2+(1-p)q_2}}\]
where $\bowtie_1, \bowtie_2 $ are two relational operators. The final postcondition indicates that, if after $P$  executes $S$ holds with probability $\bowtie_1 p_1$, and after $Q$  executes $S$ holds with probability $\bowtie_1 q_1$, then after $P  \sqcup_p Q$ executes, $S$ holds with probability $\bowtie_1 pp_1+(1-p)q_1$; The history formula can be understood similarly.

\prag{Communication Interrupt}
We define the rule for the special case $\exemptS{\langle d s = b dt + \sigma dW \& B\rangle}{ch?x}{Q}$ for simplicity,
which can be generalized to general case without any difficulty.
We use $o_F$ to denote the execution time of the SDE. The premise of the first rule indicates
that the compatible event (i.e. $h.ch!$) is not ready after the continuous terminates.
For this case, the effect of executing the whole process
is thus equivalent to that of executing the SDE.
\[
\fracN{ \begin{array}{c}
  \{A \wedge now = o; E\} \langle d s = b dt + \sigma dW \& B\rangle \{R \wedge now = o + o_F; C\} \\
       A \wedge now=o \wedge E \imply (tr=h\wedge  \rdyduring{\neg h.ch!}{[o, o+o_F]})
   \end{array} }
{ \begin{array}{l}
   \{A \wedge now = o; E\}~ \exemptS{\langle d s = b dt + \sigma dW \& B\rangle}{ch?x}{Q} ~\{ R \wedge now = o + o_F; C\}
       \end{array} }
 \]
In contrary, when the compatible event gets ready before
the continuous terminates,
the continuous will be interrupted by the communication, which is
then followed by $Q$. Thus, as shown in the following rule, the effect of executing
the whole process is equivalent to that of executing $ch?x;  Q$, plus that of
executing the $SDE$ before the communication occurs, i.e. in the first $o_1$ time units.
\[
\fracN{ \begin{array}{c}
  \{A \wedge now = o; E\}\langle d s = b dt + \sigma dW \& B\rangle \{R \wedge now = o + o_F; C\} \\
      (A \wedge now=o  \wedge E) \imply ( tr=h\wedge \rdyat{h.ch!}{(o + o_1)} \wedge o_1 \leq o_F )\\
  \triple{A \wedge B \wedge now = o; E}{ ch?x; Q}{R_1; C_1}
   \end{array} }
{ \begin{array}{l}
   \{A \wedge now = o; E\}\,  \exemptS{\langle d s = b dt + \sigma dW \& B\rangle}{ch?x}{Q} \, \\
   \qquad\qquad\qquad\qquad\qquad\{R_1;  R|_{[o, o+o_1)}\wedge C_1 \}
       \end{array} }
 \]
where $R|_{[o, o+o_1]}$ extracts from $R$ the formulas before $o+o_1$, e.g., $(P(\rdyat{S}{T})\bowtie p)|_{[o, o+o_1]}$ is equal to $P(\rdyat{S}{T})\bowtie p$ if $T$ is less or equal to $o+o_1$, and true otherwise.

\prag{Parallel Composition}

For \sm{P\|Q}, let \sm{X} be \sm{X_1 \cap X_2} where \sm{X_1 = \Sigma(P)} and \sm{X_2 = \Sigma(Q)}, then \\[1mm]
{\small \hspace*{1cm}  $\fracN{\begin{array}{c}
 A \Rightarrow A_1 \wedge A_2, \quad \triple{A_1 \wedge \now=o; E_1}{P}{R_1 \wedge tr=\trace_1\wedge \now=o_1; C_1} \\
   \triple{A_2\wedge \now=o; E_2}{Q}{R_2 \wedge tr=\trace_2\wedge \now=o_2; C_2}\\
\forall ch \in X. (\fdelta{C_1}{o_1/\now}\!\upharpoonright_{ch} \Rightarrow E_2\!\upharpoonright_{ch}) \wedge (\fdelta{C_2}{o_2/\now}\!\upharpoonright_{ch} \Rightarrow E_1\!\upharpoonright_{ch})\\
\forall dh \in X_1\setminus X.  E\!\upharpoonright_{dh} \Rightarrow E_1\!\upharpoonright_{dh} \quad \forall dh' \in X_2\setminus X.  E\!\upharpoonright_{dh'} \Rightarrow E_2\!\upharpoonright_{dh'}
\end{array}}
{\triple{A\wedge \now=o; E}{P  \| Q}{R; C'_1 \wedge C'_2}} $ } \\[1mm]
where \sm{A_1} is a property of \sm{P} (i.e., it only contains variables of \sm{P}),  \sm{A_2} a property of \sm{Q}, and
 \sm{o_1} and \sm{o_2}, \sm{\gamma_1} and \sm{\gamma_2}  logical variables representing the time and trace at  termination of \sm{P} and \sm{Q} respectively. Let \sm{o_m} be \sm{\max\{o_1, o_2\}},
\sm{R}, \sm{C'_1} and \sm{C'_2} are defined as follows: \\[1mm]
 {\small $\begin{array}{lll}
    R  &\Define& R_1[\trace_1/tr, o_1/\now] \wedge R_2[\trace_2/tr,o_2/\now] \wedge \now = o_m \wedge \trace_1\!\upharpoonright_{X} = \trace_2\!\upharpoonright_{X}\wedge tr=\trace_1 \spara \trace_2 \\[-2mm]
    C'_i &\Define&  \fdelta{C_i}{o_i/\now} \wedge \rdyduring{R_i'[o_i/\now]}{[o_i, o_m)} \mbox{  for \sm{i=1,2}}
    \end{array} $ } \\
where for \sm{i=1, 2}, \sm{R_i \Rightarrow R_i'} but \sm{tr \notin R_i'}.
At  termination of \sm{P \| Q}, the time will be
the maximum of \sm{o_1} and \sm{o_2}, and the trace will be
the alphabetized parallel of the traces of \sm{P } and
\sm{Q}, i.e. \sm{\gamma_1, \gamma_2}. In \sm{C'_1} and \sm{C'_2},
we specify that  none of variables of \sm{P} and \sm{Q} except for \sm{\now} and \sm{tr} will change after their termination.

\prag{Repetition}
For \sm{P^*}, let \sm{k} be an arbitrary non-negative integer, then ($tr \notin A$) \\[1mm]
{\small $\fracN{
\begin{array}{c}
 \{A\wedge \now=o + k*t \wedge tr=(h \cdot \alpha^k); \fdelta{E}{o/\now}\}\ P\ \\
 \qquad \qquad \qquad\{A\wedge \now=o + (k+1)*t \wedge tr=(h \cdot \alpha^{k+1}); C\}
\end{array}}{\triple{A \wedge \now=o \wedge tr=h ; E}{P^*}{A \wedge \now=o'\wedge tr = (h \cdot \alpha^*)+\tau; C \vee (\rdyat{o=o'}{\now})}} $ }\\[1mm]
 \sm{t} and \sm{\alpha} are logical variables representing the time elapsed and trace accumulated respectively by each execution of \sm{P}, and \sm{o} and \sm{o'} denote the starting and termination time
 of the loop (\sm{o'} could be infinite).
%

The general rules that are applicable to all processes, such as Monotonicity, Case Analysis, and so on, are similar to the traditional Hoare Logic. We will not list them here for page limit.

\begin{theorem}[Soundness]
 If $\vdash \triple{A; E}{P}{R; C}$, then $\models \triple{A; E}{P}{R; C}$, i.e. every theorem of the proof system is valid.
\label{Theorem:Soundness}
\end{theorem}
\begin{proof}
   The proof of this theorem can be found in Appendix.
\end{proof}
\begin{example}
  For the aircraft example, define $f(x, y)$ as $|y|$, assume $f(x_s, y_0) = |y_0| \leq \lambda p$, where $p \in [0, 1]$. Obviously, $B \rightarrow (f\geq 0) \wedge (Lf \leq 0)$ holds. By applying the inference rule of SDE, we have the following result:
  \[\triple{now=o; True}{P_{Air}}{
  \begin{array}{ll}
\exists d. now = o+d  \wedge B \wedge  P(f \geq \lambda)   s\leq p;\\
 B \wedge P(\rdyduring{f \geq \lambda  }{[o, o+d]})\leq p
  \end{array}
  }
  \]
which shows that, the probability of the aircraft entering the dangerous state is always less than or equal to $p$ during the flight. Thus, to guarantee the safety of the aircraft, $p$ should be as little as possible. For instance, if the safety factor of the aircraft is required to be $99.98 \%$, then $p$ should be less than or equal to 0.0002, and in correspondence, $|y_0| \leq \frac{\lambda}{5000}$ should be satisfied.
\end{example}

\section{Conclusion}

This paper presents stochastic HCSP (SHCSP) for modelling hybrid systems with probability and stochasticity.
SHCSP is expressive but complicated with interacting discrete, continuous and stochastic dynamics. 
We have defined the semantics of stochastic HCSP and proved that it is well-defined with respect to stochasticity. We propose an assertion language for specifying time-related and probability-related  properties of SHCSP, and have proved the measurability of it. Based on the assertion language, we define a compositional Hoare Logic for specifying and verifying SHCSP processes. The logic is an extension of traditional Hoare Logic, and can be used to reason about how the probability of a property changes with respect to the execution of a process. To illustrate our approach, we model and verify a case study on a flight planing problem at the end.

\bibliographystyle{abbrv}
\bibliography{references}

\newpage

\section*{Appendix}

\subsection{The Semantics of SHCSP}

\begin{table}
\small
\centering
\begin{align}
 &(\pskip, \rho) \leadm{\tau} (\pstop,
	\rho[tr\mapsto tr \cdot \langle \tau, now\rangle])   &\tag{Skip}
\\[1mm]\nonumber
&(\pstop, \rho) \leadm{d} (\pstop,
	\rho[\now\mapsto \now+d])   &\tag{Idle}
\\[1mm]\nonumber
&	{  (x:=e, \rho) \leadm{\tau}
	(\pstop, \rho[x\mapsto e,tr \mapsto tr \cdot \langle \tau, now\rangle ])}&\tag{Assign}
\\[1mm]\nonumber
&	{ \fracN{ \rho(\omega)(tr).ch? \not\in \rho(\omega)(\textit{rdy}) }
	{(ch?x, \rho) \leadm{\tau} (ch?x, \rho[\textit{rdy} \mapsto \textit{rdy}\cup\{tr.ch?\}])} } & \tag{In-1}
	 \\[1mm]\nonumber
&	{ \fracN{ \rho(\omega)(tr).ch? \in \rho(\omega)(\textit{rdy})}
	{(ch?x, \rho) \leadm{d} \left(ch?x,\rho[\now \mapsto \now+d], H_d^\rho
	   \right)} }& \tag{In-2}
	\\[1mm]
&	{ \fracN{\rho(\omega)(tr).ch? \in \rho(\omega)(\textit{rdy})}
	{(ch?x, \rho) \leadm{ch?b} \left(\pstop, \begin{array}{r}\rho[ \textit{rdy} \mapsto \textit{rdy}\backslash \{tr.ch?\}, x \mapsto b,
	   tr \mapsto tr\cdot \langle ch.b, now\rangle ]
	 \end{array}\right)} }& \tag{In-3}
	 \\[1mm]
	& {\fracN{ \rho(\omega)(tr).ch! \not\in \rho(\omega)(\textit{rdy}) }
	{(ch!e, \rho) \leadm{\tau} (ch!e, \rho[\textit{rdy} \mapsto \textit{rdy}\cup\{tr.ch!\}])} }& \tag{Out-1}
	 \\[1mm]
	&{ \fracN{ \rho(\omega)(tr).ch! \in \rho(\omega)(\textit{rdy})}
	{(ch!e, \rho) \leadm{d} \left(ch!e,\rho[\now \mapsto \now+d], H_d^\rho
	   \right)} }& \tag{Out-2}
	 \\[1mm]
&	 {\fracN{\rho(\omega)(tr).ch! \in \rho(\omega)(\textit{rdy})}
	{(ch!e, \rho) \leadm{ch!e} \left(\pstop, \begin{array}{r}\rho[\textit{rdy} \mapsto \textit{rdy}\backslash \{tr.ch!\},
	   tr \mapsto tr\cdot\langle ch.e, now\rangle]
	 \end{array}\right)} }& \tag{Out-3}
	 \\[1mm]
&   \fracN{
	  \begin{array}{l}
	     	     (P_1, \rho_1 )
	      \leadm{ch*} (P_1', \rho_1'),  \quad
	     (P_2,  \rho_2)
	      \leadm{\overline{ch*}} (P_2', \rho_2'),    \\
	  \end{array}}
	  {
	   (P_1 \parallel P_2, \rho_1\uplus \rho_2) \leadm{\textit{comm}(ch*, \overline{ch*})}  (P_1' \parallel P_2', \rho_1'\uplus \rho_2')}
	  &\tag{Par-1}   \\[1mm]  \nonumber
&	  \fracN{
	  \begin{array}{l}
	     (P_1, \rho_1) \leadm{\beta} (P_1', \rho_1'), \quad  \Sigma(\beta) \not \in \Sigma(P_1) \cap
	       \Sigma(P_2)
	  \end{array}}
	  {     (P_1 \parallel P_2, \rho_1\uplus \rho_2) \leadm{\beta}  (P_1' \parallel P_2, \rho_1'\uplus \rho_2)}
	 & \tag{Par-2}   \\[1mm] \nonumber
&	  {\fracN{(P_i, \rho_i) \leadm{d} (P_i', \rho_i', H_i) \mbox{ for $i =1, 2$}}
	   {(P_1 \parallel P_2, \rho_1\uplus \rho_2) \leadm{d}  (P_1' \parallel P_2', (\rho_1'\uplus \rho_2'), H_1 \uplus H_2)}}
	  &\tag{Par-3} \\[1mm] \nonumber
&	     	       	    (\pstop \parallel \pstop, \rho_1\uplus \rho_2) \leadm{\tau}  (\pstop, \rho_1\uplus \rho_2)
&	  \tag{Par-4}\\ \nonumber
%
&	 \fracN{ \rho(\omega)(B) = \ptrue}
	      { (B\rightarrow P, \rho) \leadm{\tau}
	    (P,\rho[tr\mapsto tr \cdot \langle \tau, now\rangle])}   &  \tag{  Cond-1 }  \\[1mm] \nonumber
&	\fracN{ \rho(\omega)(B) = \pfalse}
	      {(B\rightarrow P, \rho) \leadm{\tau}
	    (\pstop,\rho[tr\mapsto tr \cdot \langle \tau, now\rangle])}  &   \tag{ Cond-2}\\[1mm] \nonumber
&	\fracN{ (P,\rho) \leadm{\alpha} (P', \rho', H) \quad P' \neq \pstop}
	  {(P;Q, \rho) \leadm{\alpha} (P';Q, \rho', H)} &  ~~ \tag{ Seq-1 }  \\[1mm] \nonumber
&	\fracN{ (P,\rho) \leadm{\alpha} (\pstop, \rho', H)  }
	  {(P;Q, \rho) \leadm{\alpha} (Q, \rho', H)} & ~~ \tag{ Seq-2 }
	\nonumber\\[1mm]
&	\fracN{(P, \rho) \leadm{\alpha} (P', \rho', H) \quad P' \neq \pstop }
	        {(P^*, \rho)\leadm{\alpha} (P'; P^*, \rho', H)}   &~~ \tag{ Rep-1}  \\[1mm] \nonumber
&	        \fracN{(P, \rho) \leadm{\alpha} (\pstop, \rho', H) }
	        {(P^*, \rho)\leadm{\alpha} (P^*, \rho', H)}  & ~~ \tag{Rep-2} \nonumber \\[1mm]
	&        (P^*, \rho) \leadm{\tau} (\pstop, \rho[tr\mapsto tr \cdot \langle \tau, now\rangle]) & \tag{ Rep-3} \nonumber
\end{align}
\caption{The semantics of the rest of SHCSP}
\label{tab:osemanticsrest}
\end{table}

The semantics of the rest of SHCSP is given in Table~\ref{tab:osemanticsrest}. The semantics of \sm{\pskip} and \sm{x:=e} are defined as usual, except that for each, an internal event occurs. Rule (Idle) says that a terminated configuration can keep idle arbitrarily, and then evolves to itself. For input \sm{ch?x}, the input event has to be put in the ready set if it is enabled (In-1); then it may wait for its environment for any time \sm{d} during keeping ready (In-2); or it performs a communication and terminates, and accordingly
the corresponding event will be removed from the ready set, and \sm{x} is assigned and \sm{tr} is extended by the communication (In-3). The semantics of  output $ch!e$ is  similarly defined by rules (Out-1), (Out-2) and (Out-3).

	For \sm{P_1 \| P_2}, we always assume that the initial states \sm{\rho_1} and \sm{\rho_2} are parallelable. There are four rules: both \sm{P_1} and \sm{P_2} evolve for \sm{d} time units in case they can delay \sm{d} time units respectively; or \sm{P_1} may progress separately on internal events or external communication events (Par-2), and the symmetric case can be defined similarly (omitted here); or they together perform a synchronized communication (Par-3); or \sm{P_1\|P_2} terminates when both \sm{P_1} and \sm{P_2} terminate (Par-4). At last, the semantics for conditional, sequential, internal choice, and repetition is defined as usual.

\subsection{Proof of Theorem~\ref{welldefined}}
\myproof{
We will prove the $c\grave{a}dl\grave{a}g$, adaptedness and Markov time properties by induction on the structure of SHCSP $P$. To simplify notation, we assume that the process $P$ start at time 0 and $\Delta(P)$ is short for $\Delta(P,P')$ if $P'=\pstop$.

\begin{itemize}
\setlength{\parindent}{1.5em}

  \item Cases $\pskip$, wait $d$ and $x=e$: Deterministic times $\Delta(\pskip)=\Delta(x=e)=0$ and $\Delta(\mbox{wait }d)=d$ are trivial Markov times. For $\pskip$ and wait $d$, $H$ is adapted to the filtration generated by $\rho$. For $x=e$, $H$ is adapted to $\rho$ and $e$. For $\pskip$ and $x=e$, $H$ is trivially $c\grave{a}dl\grave{a}g$ as the time domain is \{0\}.

  \item Case In-1:  $\Delta(ch?x, ch?x)=0$ is a trivial Markov time. $H$ is $c\grave{a}dl\grave{a}g$ and adapted to the filtration generated by $\rho$.

  \item Case In-2:$\Delta(ch?x, ch?x)=d$ is a trivial Markov time. $H$ is $c\grave{a}dl\grave{a}g$ and adapted to the filtration generated by $\rho$.

  \item Case In-3:$\Delta(ch?x)=d$ is a trivial Markov time. $H$ is $c\grave{a}dl\grave{a}g$ and adapted to the filtration generated by $\rho$ and $e$.

  For cases Out-1, Out-2 and Out-3, the fact can be proved similarly.

  \item Case $\la ds=bdt+\sigma dW \& B \ra$: $\Delta(\la ds=bdt+\sigma dW \& B \ra)=\mbox{inf}\{t\geq 0: X_t \notin B\}$ is a Markov time if $B$ is any Borel set. Here, $X_t$ is the solution of  $SDE$ $ds=bdt+\sigma dW$. $H$ is adapted to the filtration generated by $(W_s)_{s\leq t}$ and $\rho$.

  \item Case $B\to P$: If $B$ is true, executing $B\to P$ is same as executing $P$. By induction hypothesis, $\Delta(P)$ is a Markov time and $H$ is $c\grave{a}dl\grave{a}g$ and adapted. If $B$ is false, the fact holds obviously.

  \item Case $P \sqcup_{p} Q$: By induction hypothesis, $\Delta(P)$ and $\Delta(Q)$ are both Markov time. So $\Delta(P \sqcup_{p} Q)$, the sum of two Markov times $p\Delta(P)$ and $(1-p)\Delta(Q)$, is also a Markov time. By induction hypothesis, $H'$ for $P$ and $H''$ for $Q$ are both $c\grave{a}dl\grave{a}g$. Because $c\grave{a}dl\grave{a}g$ functions form an algebra, $H$ is also $c\grave{a}dl\grave{a}g$ for every outcome of $\sqcup$. $H$ is adapted, because $H'$ and $H''$ are adapted and the choice $\sqcup$ generates the filtration.

  \item Case $P;Q$: Suppose $(P;Q, \rho) \leadm{\alpha} (Q, \rho', H')$ and $(Q, \rho') \leadm{\alpha} (\pstop, \rho'', H'')$. By induction hypothesis, $\Delta(P;Q,Q)=\Delta(P)$ is a Markov time and  $H'$ is $c\grave{a}dl\grave{a}g$ and adapted to $(\mathcal{F'}_t)_{t\geq 0}$. $\rho'$ is a random variable. By induction hypothesis, $\Delta(Q)$ is a Markov time and $H''$ is $c\grave{a}dl\grave{a}g$ and adapted to $(\mathcal{F''}_{t-\Delta(P)})_{t\geq \Delta(P)}$. Obviously, $\Delta(P;Q)=\Delta(P)+\Delta(Q)$ is a Markov time. $H$ is adapted to $(\mathcal{F}_t)_{t\geq 0}$, since the two parts $H'$, $H''$ are adapted. By induction hypothesis, $H$ is $c\grave{a}dl\grave{a}g$ on $[0, \Delta(P;Q,Q))$ and on $(\Delta(P;Q,Q),\infty )$, because the constituent fragments are. At $\Delta(P;Q, Q)$, $H$ is $c\grave{a}dl\grave{a}g$, by construction.

  \item Case $\langle d s = b dt + \sigma dW \& B\rangle \unrhd_d Q$: This case can be defined by $t=0; \langle d s = b dt + \sigma dW \& t<d \wedge B\rangle; t\geq d \to Q$. The fact can be proved similarly as the case $P;Q$.

  \item Case $\exempt{\langle d s = b dt + \sigma dW \& B\rangle}{i\in I}{\omega_i \cdot ch_i*}{Q_i}$: If the evolution of $SDE$ terminates before any communication occurs, this case is same as $\la ds=bdt+\sigma dW \& B \ra$. Otherwise, $H$ is $c\grave{a}dl\grave{a}g$ and adapted the filtration generated by $\rho$, $(W_s)_{s\leq t}$ and the weights $\{\omega_i\}_{i\in I}$. $\Delta(\exempt{\langle d s = b dt + \sigma dW \& B\rangle}{i\in I}{\omega_i \cdot ch_i*}{Q_i})$ is a Markov time, since the communication and $Q_i$ are both Markov times.

  \item Case $P \| Q$: Suppose $(P_1 \parallel P_2, \rho_1\uplus \rho_2) \leadm{}  (\pstop \parallel \pstop, \rho_1'\uplus \rho_2', H_1 \uplus H_2)$. Because the processes $P$ and $Q$ don't share variables, by induction hypothesis, $H=H_1 \uplus H_2$ is $c\grave{a}dl\grave{a}g$ and adapted to the filtration generated by $\rho_1\uplus \rho_2$, $(W_s)_{s\leq t}$ and the weights $\{\omega_i\}_{i\in I}$. $\Delta(P\| Q)=\mbox{max}(\Delta(P),\Delta(Q))$ is a Markov time.

\end{itemize}

}

\subsection{Proof of Theorem~\ref{Them:Mea}}
\myproof{We will prove this fact by induction on the structure of $\theta$, $S$ and $\formt$.

 $\newsemz{\theta}$ is a random variable:
\begin{enumerate}
 \item $\newsemz{c}=c$ is a random variable trivially.
 \item $\newsemz{x}=Y$ is a random variable, because $Y(\omega)=Z(\omega)(x)$ for each $\omega \in \Omega$ and $Z$ is measurable. So is $Y$.
 \item $\newsemz{f^k(E_1, ..., E_k)}=f^k(\newsemz{E_1}, ..., \newsemz{E_k})$ is a random variable, because $\newsemz{E_1},$ $ ..., \newsemz{E_k}$ are measurable and $f^k$ is Borel-measurable. Thus, the composition $f^k(\newsemz{E_1}, ..., \newsemz{E_k})$ is measurable (the $\sigma$-algebras in the composition are compatible).

The cases $\newsemz{o}$, $\newsemz{now}$, $\newsemz{u^l(T_1, ..., T_l)}$, $\newsemz{\varepsilon}$ and $\newsemz{\pair{ch.E}{T}}$ can be proved similarly.
\item $\newsemz{h_1 \cdot h_2}=\newsemz{h_1} \cdot \newsemz{h_2}$ is a product. It is also measurable by induction hypothesis (measurable functions form an algebra).
\end{enumerate}
 $\newsemz{S}$ is a random variable:
\begin{enumerate}
 \item $\newsemz{\bot}=0$ is trivially measurable.
 \item $\newsemz{h.ch?}=\mathcal{I}_{\{\omega \in \Omega |\newsemz{h}(\omega).ch? \in Z(\omega)(rdy) \}}$ is measurable, because $\newsemz{h.ch?}\equiv 0$ or 1.
 \item $\newsemz{\neg S}= 1-\newsemz{S}$ is measurable ($\newsemz{S}$ is measurable).

 $\newsemz{R^n(\termt_1, \dots, \termt_n)}$, $\newsemz{h.ch!}$ and $\newsemz{S_1 \vee S_2}$ can be proved similarly.
\end{enumerate}
$\newsemhz{\formt}$ is a random variable:
\begin{enumerate}
 \item $\newsemhz{\bot}=0$ is trivially measurable.
 \item $\newsemhz{\rdyat{S}{T}}= \Lbrack S \Rbrack^{\timefun(\newsemz{T})}$ is measurable, because $\Lbrack S \Rbrack^{\timefun(\newsemz{T})}$ is.
 \item $\newsemhz{\neg \formt} =1- \newsemhz{\formt}$ is measurable ($\newsemhz{\formt}$ is measurable).
 \item $\newsemhz{\forall v.\formt }=\mbox{inf}\{\newsemhz{\formt[b/v]} :b \in \mathbb{R} \}$ is measurable for the following reason. By Theorem 1, $\timefun$ is measurable (adapted). By induction hypothesis, $\newsemhz{\formt[b/v]}$ is measurable for each b. Consider a rational mesh $\pi:=\{b_1, b_2, \dots, b_n\}\subset \mathbb{Q}$ with $b_1\leq b_2 \leq \cdots \leq b_n$. It's obvious that $\newsemhz{\formt[b/v]}$ is measurable for each $b\in \pi$. So, the (finite) countable infimum $\mbox{inf}\{\newsemhz{\formt[b/v]} :b \in \pi \}$ is measurable. Then, the countable infimum $\mbox{inf}\{\newsemhz{\formt[b/v]} :b \in \pi \mbox{ for a rational mesh}\}$ is measurable, because the set of rational meshes is countable. Notice that $\timefun$ is $c\grave{a}dl\grave{a}g$ by Theorem 1, so $\mbox{inf}\{\newsemhz{\formt[b/v]} :b \in \mathbb{R} \}$  is measurable.

 $\newsemhz{\formt_1 \vee \formt_2}$ and $\newsemhz{\forall t.\formt }$ can be proved similarly.
\end{enumerate}
}

\subsection{Proof of Theorem~\ref{Theorem:Soundness}}
\myproof{
To prove soundness, we need to show that the axioms are valid, and that every inference rule in the proof system preserves validity. That is, if every premise of the rule is valid, then the conclusion is also valid.

We will prove the soundness theorem by induction on the structure of Stochastic HCSP processes $S$.
In the following proof, we always assume $S$ executes in parallel with its environment $E$, and
 $(S \| E, \funI_1 \uplus \funI_2) \leadm{\alpha^*} (\pstop\|E', \funI_1' \uplus \funI_2',\timefun)$; $\timefun$ is the stochastic process of the evolution and $T_0 = \rho_1(\now)$ for simplicity. Moreover, for readability, we will write \sm{\Lbrack A \Rbrack^{\funI}} and \sm{\Lbrack E \Rbrack^{\timefun,\funI}} as $\funI \models A$ and $\funI, \timefun \models E$, for any state $\funI$, any stochastic process $\timefun$, any state formula $A$, and any formula $E$.

\begin{itemize}
\setlength{\parindent}{1.5em}

  \item Case $\pskip$: The fact holds trivially from the fact $\funI_1' = \funI_1[tr+\tau]$. \\

  \item Case Assignment $x := e$: From the operational semantics, we have $\funI_1' = \funI_1[x\mapsto e,tr \mapsto tr\cdot  \pair{\tau}{\now)}]$.  Assume $\rho_1 \models (A \wedge tr=h)[e/x]$, we need to prove $\funI_1' \models A \wedge tr=h+\tau$. Obviously this holds.\\

\item Case Input $ch?x$: From the operational semantics, we have
$\funI_1' = \funI_1[\now \mapsto T_0 +d, x \mapsto b, tr \mapsto tr\cdot \pair{ch.b}{T_0+d}]$ for some $d\geq 0$ and $b$; and for any $\omega \in \Omega$ and any $t\in [T_0, T_0+d)$, $\rho_1(\omega)(tr).ch!\!\upharpoonright_{ch} \notin \timefun(t,\omega)(rdy)\!\upharpoonright_{ch}$, and
 $\rho_1(\omega)(tr).ch!\!\upharpoonright_{ch} \in \timefun(T_0+d,\omega)(rdy)\!\upharpoonright_{ch}$; and for any $t\in [T_0, T_0+d]$, $\rho_1(\omega)(tr).ch? \in \timefun(t,\omega)(rdy)$. Assume $\rho_1\models A \wedge \now = o \wedge tr=h \wedge x = v$ and $\funI_2, \timefun\models \rdyduring{\neg h.ch!}{[o,o_1)} \wedge \rdyat{h.ch!}{o_1}$, we need to prove that $\funI_1'\models A[v/x,o/\now] \wedge \now = o_1 \wedge \exists v'. (x=v' \wedge tr=h')$ and $\funI_1', \timefun\models \rdyduring{h.ch?}{[o, o_1)}$, where $h'$ is $h[v/x,o/\now] \cdot \pair{ch.v'}{o_1}$.

First from $\rho_1\models A \wedge \now = o \wedge x = v$ and the assumption that $A$ does not contain $tr$, we have $\rho_1\models  A[v/x,o/\now]$. Compare $\funI_1'$ with $\funI_1$, we can find that only variables $tr$, $\now$, and $x$ are changed. Plus that $A$ does not contain $tr$, we obtain $\funI_1'\models A[v/x,o/\now]$.

From the assumption $\funI_1, \timefun\models \rdyduring{\neg h.ch!}{[o,o_1)} \wedge \rdyat{h.ch!}{o_1}$, we can get the fact that $\forall t \in [o, o_1). \timefun(t,\cdot)(h).ch! \upharpoonright_{ch} \notin \timefun(t,\cdot)(rdy)\!\upharpoonright_{ch}$, and $\timefun(t,\cdot)(h).ch!\upharpoonright_{ch} \in \timefun(o_1,\cdot)(rdy)\!\upharpoonright_{ch}$. From $\rho_1\models  tr=h $, then $\rho(\cdot)(tr) = \rho(\cdot)(h)$, and obviously $\rho(\cdot)(h) =\timefun(t,\cdot)(h)$ since the number of $ch$ in $h$ does not change during the waiting time. Plus the fact that $T_0 = o$, we finally obtain $T_0 +d = o_1$. So  $\funI_1', \timefun\models  \now = o_1$ holds.

Denote $\rho_1'(\cdot)(x)$ by $c$, then $\funI_1'\models\exists v'. x=v'$ holds by assigning $v'$ with $c$. From the semantics of substitution, $\rho_1'(\cdot)(tr) = \rho_1(\cdot)(h) \cdot \pair{ch.c}{T_0+d}$. On the other hand, $\rho_1'(\cdot)(h[v/x,o/\now] \cdot \pair{ch.v'}{o_1}) = \rho_1(\cdot)(h) \cdot \pair{ch.c}{o_1}$. Thus, plus the above fact, we prove that $\funI_1'\models \exists v'. (x=v' \wedge tr=h')$.

Finally, from the operational rule, we have $ \funI_1', \timefun\models  \rdyduring{\rho_1(\cdot)(tr).ch?}{[T_0, T_0+d]}$. Based on the facts $T_0 = o$, $T_0 +d = o_1$, and $\rho_1(\cdot)(tr) = \rho_1(\cdot)(h)$, we  prove  the result.

\item Case Output $ch!e$: The fact can be proved similarly to $ch?x$.\\

\item Case Continuous $\langle d s = b dt + \sigma dW \& B\rangle$: First assume the continuous terminates. To prove this, we first introduce two lemmas.
\begin{lemma}
 Let $X_t$ an a.s. right continuous strong Markov process (e.g. solution from $SDE$) and $X_0=x$.  If $f\in C^2(\mathbb{R}^n,\mathbb{R}) $has compact support and $\tau$ is a Markov time with $E^x\tau<\infty$, then
\begin{displaymath}
 E^xf(X_t)=f(x)+E^x\int_0^{\tau}Af(X_s)ds
\end{displaymath}
where $Af(x):=\mathop{lim}\limits_{t\searrow 0}\frac{E^xf(X_t)-f(x)}{t}$
  \label{dynkin}
\end{lemma}
\begin{lemma}
If $f(X_t)$ is a $c\grave{a}dl\grave{a}g$ supermartingale with respect to the filtration generated by $(X_t)_{t\geq 0}$ and $f\geq 0$ on the evolution domain of $X_t$, then for all $\lambda >0$:
\begin{displaymath}
 P(\mathop{sup}\limits_{t\geq 0}f(X_t)\geq \lambda | \mathcal{F})\leq \frac{Ef(X_0)}{\lambda}
\end{displaymath}
  \label{doob}
\end{lemma}
We have $\rho_1'=\rho_1[\now \mapsto T_0+d, s \mapsto X(d,\cdot)][tr+\tau]$ for some $d\geq 0$ where $X:[0,+\infty) \times \Omega \to \mathbb{R}^{d(s)}$ is the solution of the $SDE$; and for all $t\in [T_0, T_0+d).\timefun(t,\cdot)(s) = X(t,\cdot)$. We define  another random variable  $Y=\mbox{ sup}\{f(X_t): t\in [0,d)\}$. $f\in C^2(\mathbb{R}^{d(s)},\mathbb{R})$ has compact support on $B$. Consider any $x\in \mathbb{R}^{d(s)}$ and any time $r\geq 0$. The deterministic time $r$ is a Markov time with $E^xr=r<\infty$. By Lemma 1, we have
\begin{displaymath}
 E^xf(X_r)=f(x)+E^x\int_0^{r}Af(X_t)dt
\end{displaymath}
where $Af=Lf\leq 0$ by the premise. So $\int_0^{r}Af(X_t)dt\leq 0$, hence, $E^x\int_0^{r}Af(X_t)dt\leq 0$. This implies  $E^xf(X_r)\leq f(x)$ for all $x$.

The filtration is right-continuous and $f\in C(\mathbb{R}^{d(s)}, \mathbb{R})$ is compactly supported, the strong Markov property for $X_t$ implies for all $t\geq r \geq 0$ that $E^x(f(X_t)| \mathcal{F}_r)=E^{X_r}f(X_{t-r})\leq f(X_r)$. Thus, $f(X_t)$ is a supermartingale with respect to $X_t$, because it is adapted to the filtration of $X_t$ and $E^x|f(X_t)|<\infty$ for all $t$ since $f\in C^2(\mathbb{R}^{d(s)}, \mathbb{R})$ has compact support. Consider any initial state $Y$ for $X$. By Lemma 2 and the premises, we have $P(\mathop{sup}\limits_{t\geq 0}f(X_t)\geq \lambda | \mathcal{F_0})\leq \frac{Ef(Y)}{\lambda}\leq \frac{\lambda p}{\lambda}=p$. The fact holds.

The other case is that the continuous does not terminates in finite time. From  proof above, for any $d>0$, we have $\mathcal{P}_{\leq p}(f(s)\geq \lambda )\rdyduring{}{[T_0, T_0+d]}$. So we can get $\mathcal{P}_{\leq p}(f(s)\geq \lambda )\rdyduring{}{[T_0, \infty)}$. The result holds.

\item Case Sequential Composition $P; Q$: We assume the intermediate state at termination of $P$ is $\rho_1''$ (thus
$Q$ will start from $\rho_1''[tr+\tau]$), and the behaviors
of $P$ and $Q$ are $\timefun_1$ and $\timefun_2$ respectively, whose concatenation is exactly $\timefun$.
 Assume we have $\funI_1 \models \Pre\wedge \now=o$ and $\rho_1, \timefun\models E$, we need to
prove that $\funI_1' \models \Post$ and $
\rho_1', \timefun\models \fdelta{C_1}{o_1/\now}  \wedge C$,
where $\triple{A\wedge \now=o; E}{P}{R_1 \wedge \now=o_1 \wedge tr=h_1; C_1}$ and
$\triple{R_1 \wedge \now=o_1 \wedge tr=h_1+\tau; \fdelta{E}{o/\now}}{P}{R; C}$
 as in the rule for sequential composition.

According to the inference rules, from $\triple{A\wedge \now=o; E}{P}{R_1 \wedge \now=o_1 \wedge tr=h_1; C_1}$, we can
get $\triple{A\wedge \now=o; E\!\upharpoonright_{\leq o_1}}{P}{R_1 \wedge \now=o_1 \wedge tr=h_1; C_1}$,
where $E\!\upharpoonright_{\leq o_1}$ only addresses the behavior of
environment before or equal time $o_1$. Then the proof is given as follows:
 First, from $\rho_1, \timefun\models E$, we have
 $\rho_1, \timefun_1\models E\!\upharpoonright_{\leq o_1}$, then by induction hypothesis, for $P$,
 we have  $\funI_1'' \models \Post_1 \wedge \now=o_1 \wedge tr=h_1$ and $
\rho_1'', \timefun_1 \models C_1$. Similarly, by induction hypothesis again for
$Q$, we have $\rho_1' \models R$ and $
\rho_1', \timefun_2\models C$, then $
\rho_1', \timefun \models C$. From $
\rho_1'', \timefun_1 \models C_1$, we have $
\rho_1', \timefun \models C_1[o_1/\now]$. The result is proved finally.

\item Case Probabilistic Choice $P  \sqcup_p Q$: We may assume $\bowtie$ is $\geq$. From operational semantics, we have $\triple{\mathcal{P}_{\geq p'}(S); E}{P}{\mathcal{P}_{\geq p_1}(S); C_1} $ with probability $p$ and $\triple{\mathcal{P}_{\geq p'}(S); E}{Q}{\mathcal{P}_{\geq p_2}(S); C_2}$ with probability $1-p$. Assume $\funI_1 \models \Pre $, and $\funI_2, \timefun \models E$. By the law of total probability, we can easily get  $\rho_1' \models \mathcal{P}_{\geq pp_1+(1-p)p_2}(S)$ and $\timefun \models C_1 \vee C_2$.

\item Case Communication Interrupt: Assume $\funI_1 \models \Pre \wedge \now=o$, and $\funI_2, \timefun \models E$.
 For the first case, assume we have
$ \{A \wedge \now = o; E\} \langle d s = b dt + \sigma dW \& B\rangle \{R \wedge \now = o + o_F; C\}$, and
$ (A \wedge \now=o \wedge E) \imply (tr=h\wedge  \rdyduring{\neg h.ch!}{[o, o+o_F]})$, we need to prove
$\rho_1' \models R \wedge \now = o + o_F$ and $\funI_1', \timefun \models C$. From the assumption,
we have $\funI_1 \models tr=h$ and $\funI_2, \timefun \models\rdyduring{\neg h.ch!}{[o, o+o_F]}$.
According to the operational semantics, the final state and the behavior of interrupt are equal to the ones of continuous. The result holds by induction hypothesis.

For the second case, assume we have
$\{A \wedge \now = o; E\} \langle d s = b dt + \sigma dW \& B\rangle\{R \wedge \now = o + o_F; C\} $,
 $   (A \wedge \now=o  \wedge E) \imply ( tr=h\wedge \rdyat{h.ch!}{(o + o_1)} \wedge
           o_1 \leq o_F )$, and
 $ \triple{A \wedge \now = o; E}{ ch?x; Q}{R_1; C_1}$, we need to prove
 $\rho_1' \models R_1$ and $\funI_1', \timefun \models \rdyduring{(\mathcal{P}_{\leq p}(f(s)\geq \lambda ) \wedge B)}{(o, o+o_1)}\wedge C_1$.
 From the assumption,
we have $\funI_1 \models tr=h$ and $\funI_2, \timefun \models\wedge \rdyduring{\neg h.ch!}{[o, o+o_1)} \wedge \rdyat{h.ch!}{(o + o_1)} \wedge
           o_1 \leq o_F $.
According to the operational semantics, the final state and the behavior of interrupt are equal to the ones of $ch?x; Q$, but in the
first $o_1$ time units, the continuous is also executing.
The result also holds by induction hypothesis.

\item Case Parallel Composition $P\|Q$: From the operational semantics, there must exist
$\rho_{11}$ and $\rho_{11}'$ , $\rho_{12}$ and $\rho_{12}'$ for initial states and terminating states of $P$ and $Q$ respectively, which satisfy: $\rho_1 = \rho_{11} \uplus \rho_{12}$ and $\rho_1' = \rho_{11}' \uplus \rho_{12}'$;
$\rho_{11}'(\cdot)(tr)\!\upharpoonright_X = \rho_{12}'\!\upharpoonright_X$ (assuming $P$ and $Q$ terminate at the same time here, which will be generalized in the following proof).
Assume we have $\funI_1
\models \Pre \wedge \now=o$, and $\rho_2, \timefun \models E$,
we need to prove $\funI_1' \models R$ and $\rho_1', \timefun \models C'_1 \wedge C'_2$, where  $\triple{A_1 \wedge \now=o; E_1}{P}{R_1 \wedge tr=\gamma_1 \wedge \now=o_1; C_1} $ and $\triple{A_2 \wedge \now=o; E_2}{Q}{R_2 \wedge tr=\gamma_2 \wedge\now=o_2; C_2}$ hold; and compatibility check $\forall ch \in X. (\fdelta{C_1}{o_1/\now}\!\upharpoonright_{ch} \Rightarrow E_2\!\upharpoonright_{ch}) \wedge (\fdelta{C_2}{o_2/\now}\!\upharpoonright_{ch} \Rightarrow E_1\!\upharpoonright_{ch})$,
$\forall dh \in X_1\setminus X.  E\!\upharpoonright_{dh} \Rightarrow E_1\!\upharpoonright_{dh}$, and $ \forall dh' \in X_2\setminus X.  E\!\upharpoonright_{dh'} \Rightarrow E_2\!\upharpoonright_{dh'}$ hold.
Among them,  $R$, $C'_1$ and $C'_2$ are defined as in the rule  for parallel composition.
The proof is given by the following steps.

First of all, we prove that $\rho_{11}', \timefun \models C_1 $ and
$\rho_{12}', \timefun  \models C_2$. If they do not hold, assume $C_1$
fails to hold  not later than $C_2$, and the first time for which
$C_1$ does not hold is $t_1$ (when it exists), then for all $t <
t_1$, $C_2$ holds. There are three kinds of formulas at time
$t_1$ in $C_1$: if the formula is for internal variables or
internal communication (between $P$ and $Q$) non-readiness, then it
will not depend on $Q$ or $E$, according to the
fact that $C_1$ holds before time $t_1$, it must hold at $t_1$; if
the formula is for external communication readiness, first from compatibility check, for any channel
$dh \in X_1 \setminus X$, it does not occur in $C_2$, then
we have $E\!\upharpoonright_{dh} \Rightarrow E_1\!\upharpoonright_{dh}$, where $E\!\upharpoonright_{dh}$  extracts formulas related to
communications along $dh$ from $E$.
Then from
$\rho_2, \timefun \models E$, we have
  $\rho_2, \timefun \models {E_1}\!\upharpoonright_{dh}$, and thus
$\rho_{12} \uplus \rho_2, \timefun \models  {E_1}\!\upharpoonright_{dh}$.
By induction hypothesis, the formula considered must hold at $t_1$; if the formula is for
internal communication readiness, then there must exist an open
interval $(t_0, t_1)$ during which it is not satisfied. From the
assumption, $C_2$ holds in the interval $(t_0, t_1)$, thus
$E_1\!\upharpoonright_X$ holds in the interval $(t_0, t_1)$. By induction, the
internal communication readiness assertions in $C_1$ hold in the
interval $(t_0, t_1)$. We thus get a contradiction. Therefore, we
can get the fact that, both  $\rho_{11}', \timefun \models C_1 $ and
$\rho_{12}', \timefun  \models C_2$  hold.
On the other hand, if such $t_1$ does not exist, there must exist an
open interval $(t_2, t_3)$ such that for all $t\leq t_2$, $C_1$ and
$C_2$ hold, while $C_1$ does not hold in $(t_2, t_3)$. The proof is
very similar to the above case. We omit it here for avoiding
repetition.

Based on the above facts, from $\funI_1 \models \Pre_1$ and $\rho_1, \timefun\models E$,
and compatibility check, we have therefore $\rho_{12}\uplus\rho_2, \timefun\models E_1$.
Similarly, we can get for another process $Q$ that $\rho_{12} \models A_1 \wedge \now=o$, and
  $\rho_{11}\uplus\rho_2, \timefun\models E_2$. Then, by induction on $P$ and $Q$, we have
  $\rho_{11}' \models R_1 \wedge tr=\gamma_1 \wedge \now=o_1$ and $\rho_{11}', \timefun \models C_1$;
  $\rho_{12}' \models R_2 \wedge tr=\gamma_2 \wedge\now=o_2$ and $\rho_{12}', \timefun \models C_2$ respectively.

Notice that $\rho_{11}' \uplus \rho_{12}'$, i.e. $\rho_1'$, only redefines the values of
  $tr$ and $\now$, where the communications are arranged in the order according
  to their occurring time, and variable $\now$ takes the greater value between $\rho_{11}'(\cdot)(\now)$ and
  $\rho_{12}'(\cdot)(\now)$.
 Obviously, we have $\rho_1' \models \Post_1[\gamma_1/tr,o_1/\now]
\wedge \Post_2[\gamma_2/tr,o_2/\now] \wedge \now = o_m$. And, $\rho_1' \models \gamma_1\!\upharpoonright_{X}=\gamma_2\!\upharpoonright_{X}$ holds because of synchronization.
From the definition of $\uplus$, $\rho_1'(tr)(t) \in
  \funI_{11}'(tr)(t) \|  \funI_{12}'(tr)(t)$,
  we can easily get the fact $\rho_1' \models tr=\gamma_1 \spara \gamma_2$.
  Thus $R$ holds for the final state.

From $\rho_{11}', \timefun \models C_1 $ and
$\rho_{12}', \timefun  \models C_2$, considering that only $\now$ change and matter,
we have $\rho_1' , \timefun \models \fdelta{C_1}{o_1/\now}
\wedge \fdelta{C_2}{o_2/\now}$.
After $P$ or $Q$ terminates, only $rdy$, $tr$ and $\now$ may change, plus the fact that $R_1$ and $R_2$ do not
contain readiness, $R_1 \Rightarrow R_1'$, $R_2 \Rightarrow R_2'$, and $R_1', R_2'$ do not contain $tr$,
we have $\rho_1' , \timefun \models \rdyduring{R_1'[o_1/\now]}{[o_1, \now)}$ and
$\rho_1' , \timefun \models \rdyduring{R_2'[o_2/\now]}{[o_2, \now)}$.
The whole result is proved.

\item Case Repetition $P^*$: From the operational semantics, we have
there must exist a finite integer $n>0$, and $\funI_{11},
 ..., \funI_{1n}$ such that $(P^* \| E, \funI_{11} \uplus \funI_2) \leadm{\alpha^*} (\pstop; P^*\|E_1, \funI_{12} \uplus \funI_{21})
 \ldots  \leadm{\alpha^*} (P^*\|E', \funI_{1n} \uplus \funI_2') \leadm{\tau} (\pstop\|E', \funI_{1n}[tr+\tau] \uplus \funI_2')$ where $\funI_{11} = \funI_1, \funI_{1n}[tr+\tau] = \funI_1'$.
Assume $\funI_1 \models A \wedge \now=o \wedge tr=h$ and $\funI_2, \timefun \models
E$, we need to prove that
$\funI_1' \models A \wedge \now=o' \wedge tr = h \cdot w^* + \tau$ and $\funI_1', \timefun \models C \vee (\rdyat{o=o'}{\now})$, where
$\triple{A\wedge \now=o + k*t \wedge tr=h \cdot w^k; E[o/\now]}{P}{A \wedge \now=o + (k+1)*t \wedge tr=h \cdot w^{k+1}; C}$ holds as defined in the rule for
Repetition for any non-negative integer $k$.

If $n=1$, then we have $ \funI_1[tr+\tau] =  \funI_1'$, let $o = o'$, the fact holds directly. If
$n>1$, from $\funI_1 \models A \wedge \now=o \wedge tr=h$ and $\funI_2, \timefun \models
E[o/\now]$, then let $k$ be $0$, by induction hypothesis, we have $\funI_{12} \models A \wedge \now= o'\wedge tr=h \cdot w$ by assigning
$o'$ by $o + t $, and $\rho_{12}, \timefun \models C$. Recursively repeating the proof, plus the fact for any $k$, $\funI_{1k}, \timefun \models
E[o/\now]$, we can prove the result.
\end{itemize}
}

\end{document}